\documentclass[12pt]{amsart}

\usepackage{vmargin}            
\usepackage[english]{babel}
\usepackage[utf8]{inputenc}
\usepackage[T1]{fontenc}
\usepackage{enumerate}
\usepackage{amsmath}
\usepackage{hyperref}
\usepackage{braket}
\usepackage{pdfsync}

\newtheorem{theorem}{Theorem}[section]
\newtheorem{lemma}[theorem]{Lemma}
\newtheorem{proposition}[theorem]{Proposition}
\newtheorem{corollary}[theorem]{Corollary}
\theoremstyle{definition}

\theoremstyle{definition}

\newtheorem{remark}[theorem]{Remark}

\newtheorem*{remark*}{Remark}
\numberwithin{equation}{section}
\newenvironment{proof-sketch}{\textit{Sketch of proof.~}}{\quad}

\newcommand{\R}{\mathbb{R}}
\newcommand{\C}{\mathbb{C}}
\newcommand{\N}{\mathbb{N}}
\newcommand{\Z}{\mathbb{Z}}

\newcommand{\n}{\mathbf{n}}
\newcommand{\vecteur}{\mathbf{v}}
\newcommand{\x}{\mathbf{x}}

\newcommand{\pa}{\partial}
\newcommand{\eps}{\varepsilon}
\renewcommand{\Re}{\mathrm{Re}}

\DeclareMathOperator{\RE}{Re}

\DeclareMathOperator{\ran}{ran}
\DeclareMathOperator{\Span}{span}
\DeclareMathOperator{\Sp}{Sp}

\begin{document}


\title[Dirac Operators with Infinite Mass Boundary Conditions in Sectors]{Self-Adjointness of Dirac Operators with\\ Infinite Mass Boundary Conditions in Sectors}
\author[L. Le Treust]{Lo\"{i}c Le Treust}

\address{Aix Marseille Univ\\ CNRS\\ Centrale Marseille, I2M\\ Marseille, France}
\email{loic.le-treust@univ-amu.fr}

\author[T. Ourmi\`{e}res-Bonafos]{Thomas Ourmi\`{e}res-Bonafos}

\address{Laboratoire Math\'ematiques d'Orsay\\ Univ. Paris-Sud\\ CNRS\\ Universit\'e Paris-Saclay\\ 91405 Orsay, France}
\email{thomas.ourmieres-bonafos@math.u-psud.fr}
%

%
\begin{abstract}
	This paper deals with the study of the two-dimensional Dirac operator with {\it infinite mass} boundary conditions in sectors. We investigate the question of self-adjointness depending on the aperture of the sector: when the sector is convex it is self-adjoint on a usual Sobolev space whereas when the sector is non-convex it has a family of self-adjoint extensions parametrized by a complex number of the unit circle. As a byproduct of the analysis we are able to give self-adjointness results on polygonal domains. We also discuss the question of {\it distinguished} self-adjoint extensions and study basic spectral properties of the Dirac operator with a mass term in the sector.
\end{abstract}
\maketitle
%

\section{Introduction}
	\subsection{Motivations and state of the art}
	Relativistic quantum particles (electrons or quarks) confined in planar or spatial regions are efficiently described by a Hamiltonian given by the Dirac operator in a domain of the two- or three-dimensional Euclidean space with adequate boundary conditions. The question we address in this paper is related to two models of mathematical physics involving such Hamiltonians: the so-called \emph{graphene quantum dots} and the \emph{MIT bag model}. We discuss both of them below.

	\paragraph{Graphene quantum dots}	
	These two-dimensional models come into play when investigating graphene, that is a two-dimensional allotrope of carbon where the atoms are located on an infinite hexagonal lattice (see, for instance, \cite{RevModPhys.81.109}). It turns out that the massless Dirac operator is the effective Hamiltonian describing low-energy properties of electrons in such a structure but, as in practice only finite size sheets of graphene can be obtained, one has to consider the Dirac operator in a bounded domain. The shape of this domain may varies, according to whether one is interested in nano-tubes, nano-ribbons or flakes and the bounded confining systems are called graphene quantum dots. Mathematically, this translates to the study of the massless Dirac operator imposing particular boundary conditions on the boundary of a domain (the quantum dot). On the one hand, boundary conditions can be obtained by specific cuts of the carbon sheet, the most common ones being the \emph{zigzag} and \emph{armchair} boundary conditions (see \cite{PhysRevB.77.085423}). On the other hand, another confining system can be formally obtained \emph{via} a coupling of the massless Dirac operator with a mass potential that is zero inside the quantum dot and infinite elsewhere (see for instance \cite{PhysRevB.78.195427} where this boundary condition is physically justifed and \cite{2016arXiv160309657S} for a rigorous mathematical derivation). The resulting boundary conditions are called \emph{infinite mass} boundary conditions and the two-dimensional Dirac operator in a domain with these precise boundary conditions is the operator we are interested in in this paper.
	\paragraph{MIT bag model}
In the mid-70's, physicists in the MIT proposed a phenomenological model to describe the confinement of quarks inside hadrons (see \cite{MIT061974,MIT101974,MIT101975,hosaka2001quarks,johnson}) and this model, called the MIT bag model, has predicted successfully many properties of hadrons (see, for instance, \cite{PhysRevD.12.2060}). It involves the three-dimensional Dirac operator with a mass term in a bounded domain of the Euclidean space with adequate boundary conditions. These conditions can be seen as the three-dimensional counterpart of the infinite mass boundary conditions for graphene quantum dots and our interest in this three-dimensional model has drawn our attention to its two-dimensional analogue.

	From a mathematical point of view, the first challenge studying Dirac operators in bounded domains is to understand on which domain they are self-adjoint. Because the Dirac operator is an elliptic operator of order one, one expect this domain to be contained in the usual Sobolev space $H^1$. Of course, it depends on the boundary conditions and it is true for the MIT bag model for sufficiently smooth domains, as proved in \cite{arrizabalaga:hal-01540149} for $\mathcal{C}^{2,1}$-smooth domains and in \cite{2016arXiv161207058O} for $\mathcal{C}^2$-smooth domains. Moreover, when one deals with $\mathcal{C}^\infty$-smooth domains more general results can be found in \cite[Thm. 4.11]{MR3618047} and in \cite{MR2536846} where the authors use pseudo-differential techniques and Calder\'on projectors. In dimension two, the question of self-adjointness is addressed in \cite{BFSV17} for $\mathcal{C}^2$-smooth domains of the Euclidean plane. For a large class of boundary conditions, the authors prove that indeed, the domain of self-adjointness is contained in $H^1$. However, we would like to point out that it is known to be false for zigzag boundary conditions (see \cite{FS14,Sch95}) and it has important consequences for the spectral features of the problem.

	In this paper, we tackle the question of self-adjointness for the two-dimensional Dirac operator with infinite mass boundary conditions in sectors and, as a byproduct of the analysis, we deduce a self-adjointness result on polygonal domains. To our knowledge, this is the first attempt to mathematically handle this question for corner domains although \emph{polygonal graphene quantum dots} have drawn attention among the physicists community in the past few years (see for instance \cite{PhysRevB.77.193410,PhysRevB.81.033403,0957-4484-19-43-435401,PhysRevB.81.085432} for triangular, rectangular, trapezoidal and hexagonal graphene quantum dots). Actually, the question of self-adjointness is the first step toward future investigations. First, in a perspective of numerical applications, it is rather natural to consider polygonal domains because any two-dimensional domain, even smooth, is meshed with polygons. Second, we have in mind the investigation of the MIT bag operator in polyhedral domains in the regime of infinite mass. This is motivated by the work \cite{MR3663620} where it is proved that for smooth domains, the asymptotics of the eigenvalues in the usual Dirac gap are driven by a Laplace-Beltrami operator on the boundary with a curvature induced potential. As corners can be thought of as points of ``infinite'' curvature we aim to understand their influence on the spectrum of the MIT bag operator in this asymptotic regime. Because the geometry is less involved in dimension two, in a first attempt to shed some light on this question, we focus on the two-dimensional counterpart of the MIT bag model that is the Dirac operator with a mass term and infinite mass boundary conditions. This motivates the part of the present paper concerning basic spectral properties of such an operator.

	Finally, let us describe the techniques we use in this paper. They are reminiscent of \cite[Section 4.6]{Thaller1992} and \cite{esteban2017domains} where the three-dimensional Dirac-Coulomb operator is studied as well as \cite{MR1025228} which deals with the case of a radial $\delta$-shell interaction. The key point in all these works is to investigate the restrictions of the operator to stable subspaces of functions of fixed angular momentum. Then, the restricted operators only act in the radial variable and their self-adjoint extensions can be studied using classic ODE techniques \cite{MR923320}. We obtain the result for the whole operator using the standard result \cite[Theorem X.11]{reed-2}.
	\subsection{The Dirac operator with infinite mass boundary conditions in sectors}
	For $\omega\in(0,\pi)$, let $\Omega_\omega$ denote the two-dimensional sector of half-aperture $\omega$
	\begin{equation}\label{def:sector}
		\Omega_\omega = \{
			r(\cos(\theta),\sin(\theta))\in \R^2\,: r>0, \; |\theta|<\omega
		\}\,.
	\end{equation}
	Let $(D, \mathcal{D}(D))$ denote the Dirac operator with mass $m\in\mathbb{R}$ and infinite mass boundary conditions in $\Omega_\omega$. It is defined by
	\begin{equation}\label{def:DiracMIT}
		\begin{split}
			&
			\mathcal{D}(D) = 
				\{
					u\in H^1({\Omega_\omega}\, ;\ \C^2): \mathcal{B}_{\n} u = u \mbox{ on } \pa \Omega_\omega
				\},
			\\
			&
			D u = -i\sigma\cdot \nabla u+m\sigma_3 u, \mbox{ for all }u \in \mathcal{D}(D),
		\end{split}
	\end{equation}
	%
	%
	where the Pauli matrices $\sigma = (\sigma_1,\sigma_2,\sigma_3)$ are $2\times2$ hermitian matrices defined by
	\[
		\sigma_1 = 
		\begin{pmatrix}
			0&1\\1&0
		\end{pmatrix}
		,\,
		\sigma_2 = 
		\begin{pmatrix}
			0&-i\\i&0
		\end{pmatrix}
		\mbox{ and }
		\sigma_3 = 
		\begin{pmatrix}
			1&0\\0&-1
		\end{pmatrix}.
	\]
	For $a\in\mathbb{R}^d$ (for $d=2,3$), we set
	\[
		\sigma\cdot a= \sum_{k=1}^d\sigma_k a_k.
	\]
	Remark that the Pauli matrices satisfy
	\begin{equation}\label{eq:DiracMatrixProp}
		(\sigma\cdot a)(\sigma\cdot b) = 1_2(a\cdot b) + i\sigma\cdot (a\times b),\quad \text{for } a,b \in \mathbb{R}^3.
	\end{equation}
	 For almost all $s\in \pa \Omega_\omega$, $\n(s)$ denotes the outer unit normal at the point $s$. Let $\vecteur\in\mathbb{R}^2$ be a unit vector, the matrix $\mathcal{B}_{\vecteur} $ is defined by
	\begin{equation}\label{def:BoundaryCond}
		\mathcal{B}_\vecteur = -i\sigma_3\sigma\cdot \vecteur.
	\end{equation}
	Let us remark that $\mathcal{B}_\vecteur$ satisfies
	\begin{equation}\label{eq:propB}
		{\mathcal{B}_\vecteur}^* = \mathcal{B}_\vecteur,
		\quad
		{\mathcal{B}_\vecteur}^2 = 1_2, 
		\quad
		\Sp (\mathcal{B}_\vecteur) = \{\pm 1\},
	\end{equation}
	where $1_2$ denotes the $2\times2$ identity matrix.
	\begin{remark}
		The operator $(D, \mathcal{D}(D))$ is symmetric and densely defined (see Lemma \ref{lem:sym}).
	\end{remark}
	\subsection{Main results}
	Our main result is stated in \S \ref{ref:para-1} and concerns the question of self-\break adjointness of the operator $(D,\mathcal{D}(D))$ in a sector. When there are several self-adjoint extensions, we discuss in \S \ref{ref:para-2} which one should be chosen as the ``distinguished'' one. Finally in \S \ref{ref:para-3} we state results regarding polygonal domains and in \S \ref{ref:para-4} we give basic spectral properties of $D$.
	\subsubsection{Self-adjointness in sectors}\label{ref:para-1}
	In the following theorem, we give all self-adjoint extensions of the Dirac operator with infinite mass boundary conditions in sectors.
	\begin{theorem}\label{thm:sa}
		The following holds.
		\begin{enumerate}[(i)]
			\item \label{thm:cc} \textbf{[Convex sectors]} 
			
			If $\omega\in(0,\pi/2]$, $(D, \mathcal{D}(D))$ is self-adjoint.
			\item\label{thm:ncc} \textbf{[Non-Convex sectors]}
			
			 If $\omega\in(\pi/2,\pi),$ $(D, \mathcal{D}(D))$ is symmetric and closed but not self-adjoint. The set of self-adjoint extensions of $D$ is the collection of operators 
			 \[\{(D^\gamma, \mathcal{D}(D^\gamma))| \,\gamma\in \C, |\gamma| = 1\}\]
			 defined for $v\in\mathcal{D}(D) $ by
				\[
					\begin{split}
						&\mathcal{D}(D^\gamma)
						= \mathcal{D}(D) 
						+ \Span (v_++\gamma v_-),
						\\
						&D^\gamma v = D v,
						\\
						&D^\gamma(v_++\gamma v_-)
						 = i(v_+-\gamma v_-) + m\sigma_3(v_++\gamma v_-),
					\end{split}
				\]
			and where
			\[	
			\begin{split}
				&v_+(r\cos(\theta),r\sin(\theta)) = K_{\nu_0}(r)u_0(\theta) -  i K_{\nu_0+1}(r)u_{-1}(\theta),
				\\
				&v_-(r\cos(\theta),r\sin(\theta)) = K_{\nu_0}(r)u_0(\theta) +  i K_{\nu_0+1}(r)u_{-1}(\theta),
			\end{split}
			\]
		 \[
			u_{0}(\theta) :=\frac{1}{2\sqrt{\omega}}
				\begin{pmatrix}
				e^{i\theta\nu_0}\\
				-ie^{-i\theta\nu_0}
			\end{pmatrix},
			u_{-1}(\theta) :=\frac{1}{2\sqrt{\omega}}
				\begin{pmatrix}
				e^{-i\theta(\nu_0+1)}\\
				ie^{i\theta(\nu_0+1)}
			\end{pmatrix}.
		 \]
		Here, $r>0$, $\theta\in(-\omega,\omega)$,
			$\nu_0 = \frac{\pi-2\omega}{4\omega}$ and $K_\nu$ denotes the modified Bessel function of the second kind of parameter $\nu$.
		\end{enumerate}
	\end{theorem}
	\begin{remark} The distinction between convex an non-convex sectors in Theorem \ref{thm:sa} is not surprising: it is reminescent of \cite{MR0226187} where the so-called corner singularities for elliptic operators of even order are investigated. We also mention the books \cite{MR1173209,grisvard1980boundary} where the Laplacian in polygonal domains with various boundary conditions is studied.
	\end{remark}
	\begin{remark}\label{rem:turn}
		For $\theta_0\in[0,2\pi]$, let us consider the rotated sector
		\[
			\Omega_{\omega,\theta_0} := \{r\big(\cos(\theta),\sin(\theta)\big)\in\mathbb{R}^2 :  r>0, |\theta-\theta_0|<\omega\}.
		\]
		Remark that $e^{-i\sigma_2\theta_0}$ is a rotation matrix of angle $\theta_0$ and we have $\Omega_{\omega,\theta_0} = e^{-i\sigma_2\theta_0} \Omega_\omega$.		

		Let $\mathcal{U}_{\theta_0}$ be the unitary transformation defined by
		\[
			\begin{array}{lll}
				\mathcal{U}_{\theta_0}:\,
				& L^2(\Omega_{\omega,\theta_0},\, \C^2)&\longrightarrow L^2(\Omega_{\omega},\, \C^2)\\
				&  v&\longmapsto e^{i(\theta_0/2)\sigma_3} v(e^{-i\sigma_2\theta_0} \cdot ).
			\end{array}
		\]
		It satisfies
		\[
			\begin{split}
				&
				\mathcal{U}_{\theta_0}^{-1}(-i\sigma\cdot \nabla +m\sigma_3)\mathcal{U}_{\theta_0}	= -i\sigma\cdot \nabla +m\sigma_3,
				\\
				& \mathcal{U}_{\theta_0}^{-1}(\sigma\cdot \n)\mathcal{U}_{\theta_0} = \sigma\cdot \left(e^{-i\sigma_2\theta_0}\n\right),
			\end{split}
		\]
		for all unit vector $\n\in \R^2$ (see  \cite[Sections 2 and 3]{Thaller1992}). This ensures that Theorem \ref{thm:sa} essentially covers every sectors.
	\end{remark}
	\begin{remark}\label{rem:besselasym}
		%
		Let $\nu\in \R$. For further use, we recall some properties of the modified Bessel functions $K_\nu$ of the second kind (see \cite[Chapter 7 Section 8 and Chapter 12 Section 1]{olver2014asymptotics} or \cite{abramowitz1964handbook}).
		\begin{enumerate}[(i)]
			\item The functions $r\in (0,+\infty)\mapsto K_\nu(r)\in\R$ are positive and decreasing.
			\item For $r>0$, we have 
				\[
					K_{\nu}(r) = K_{-\nu}(r).
				\]
			\item For $r>0$, we have 
		\begin{equation}\label{eq:besselasym}
			K_\nu(r) \sim_{r\to 0}
			\begin{cases}
				\frac{\Gamma(\nu)}{2}\left(\frac{r}{2}\right)^{-\nu}
				&\mbox{ if $\nu>0$}
				\\
				-\log(r)
				&\mbox{ if $\nu=0$}
			\end{cases}
		\end{equation}
		and
		\[
			K_\nu(r) \sim_{r\to +\infty}
				\left(\frac{\pi}{2r}\right)^{1/2}e^{-r}.
		\]
		\end{enumerate}
		In particular, the domain of $D^\gamma$ (see Theorem \ref{thm:sa} \eqref{thm:ncc}) rewrites using\break 
		$
			r^{-|\nu_0|}\chi(r)
		$
		and
		$
			r^{-(1-|\nu_0|)}\chi(r)
		$,	
		instead of $K_{\nu_0}$ and $K_{\nu_0+1}$, respectively. Here $\chi:\R_+\mapsto [0,1]$ is a smooth function which equals $1$ in a neighborhood of $0$ and $0$ for $r$ large enough.
	\end{remark}
	\subsubsection{Physical remarks on the self-adjoint extensions in sectors}\label{ref:para-2}
	For non-convex sectors, a natural question is to know whether some self-adjoint extensions given in Theorem \ref{thm:sa}\eqref{thm:ncc} are more relevant than others from the physical point of view. The following propositions try to shed some light on this question.
	\vfill\eject
	
	\par
	\emph{Charge conjugation symmetry.}
	
	The Dirac operator anticommutes with the charge conjugation operator $C$. It is defined for $u\in\C^2$ by
	\begin{equation}\label{eqn:def_chargeconj}
		C u = \sigma_1 \overline{u}.
	\end{equation}
	In particular, for all $\omega\in(0,\pi)$, the operator $C$ is an antiunitary transformation that leaves $\mathcal{D}(D)$ invariant, it satisfies $C^2 = 1_2$ and
	\[
		D C = -C D.
	\]
	This property is strongly related to the particle/antiparticle interpretation of the spectrum of the Dirac operator (see \cite[Section 1.4.6]{Thaller1992}).
	The following proposition gives the extensions of $D$ that still satisfy these properties with respect to the charge conjugation operator $C$.
	\begin{proposition}\label{prop:charge_conj} 	%
	Let $\omega\in(\pi/2,\pi)$. The only self-adjoint extensions of\break $(D,\mathcal{D}(D))$ such that
	\[
			C\mathcal{D}(D^\gamma)= \mathcal{D}(D^\gamma)
	\]
	are the extensions $\big(D^\gamma,\mathcal{D}(D^\gamma)\big)$ for $\gamma=\pm 1$. In these cases, we have the anticommutation relation
	\[
		\{C,D^\gamma\} = CD^\gamma + D^\gamma C = 0.
	\]
	\end{proposition}
	\par
	\emph{Scale invariance.}

	Since $\Omega_\omega$ is invariant by dilations, we immediately get that $\mathcal{D}(D)$ is stable under the action of the group of dilations. For non-convex sectors, we obtain the following proposition.
	\begin{proposition}\label{prop:scaling}
		Let $\omega\in(\pi/2,\pi)$. The only self-adjoint extensions of\break $(D,\mathcal{D}(D))$ such that for all $u\in \mathcal{D}(D^\gamma)$ and all $\alpha>0$ we have
	\[
		[\x\in\Omega_\omega\mapsto u(\alpha\x)\in \C^2]\in \mathcal{D}(D^\gamma)
	\]
	are the extensions $\big(D^\gamma,\mathcal{D}(D^\gamma)\big)$ for $\gamma=\pm 1$. 	\end{proposition}
	This proposition is essential in the proofs using Virial identities (see Remark \ref{rem:virial}). 
	\par
	\emph{Kinetic energy.}

	From a physical point of view, it is reasonable to impose the domain of the Dirac operator with infinite mass boundary conditions to be contained in the \emph{formal} form domain $H^{1/2}(\Omega_\omega)$. It turns out that only a single self-adjoint extension of $D$ satisfies this condition.
	\begin{proposition}\label{prop:dist_ext} Let $\omega\in(\pi/2,\pi)$. The only self-adjoint extension of\break $(D,\mathcal{D}(D))$ satisfying $\mathcal{D}(D^\gamma)\subset H^{1/2}(\Omega_\omega)$ is $(D^{1},\mathcal{D}(D^{1}))$.
	\end{proposition}
	\begin{remark}\label{rem:regudom}
		The proof of Proposition \ref{prop:dist_ext} shows a stronger statement. Indeed, if $\gamma=1$, we have $\mathcal{D}(D^\gamma)\subset H^{3/4-\varepsilon}(\Omega_\omega)$ for all $\varepsilon\in(0,1/4)$. 
	\end{remark}
	\subsubsection{About polygonal domains}\label{ref:para-3}
		Using Theorem \ref{thm:sa}\eqref{thm:cc}, Remark \ref{rem:turn} and partitions of unity, we obtain the following result.
	\begin{corollary}
		Let $\Omega\subset \R^2$ be a convex polygonal domain. The Dirac operator $(D_\Omega,\mathcal{D}(D_\Omega))$ defined by
		\[
			\begin{split}
				&\mathcal{D}(D_\Omega) = \{u\in H^1(\Omega,\C^2),\ \mathcal{B}_\n u = u \mbox{ on }\pa \Omega\},
				\\
				&Du = -i\sigma\cdot\nabla u + m\sigma_3u\mbox{ for all } u\in \mathcal{D}(D_\Omega),
			\end{split}
		\]
		is self-adjoint.
	\end{corollary}
	\begin{remark}
		A similar statement can be formulated for non-convex polygonal domains using Theorem \ref{thm:sa}\eqref{thm:ncc}. We chose not to write it down here for the sake of readability.
	\end{remark}
	\subsubsection{Spectral properties in sectors}\label{ref:para-4}
	Now, we investigate spectral properties of the self-adjoint operators in sectors. We restrict ourselves to the \emph{physical case} $\gamma = 1$ and, for the sake of readability, we introduce the following unified notation:
	\[
		D^{sa} = \begin{cases}
			D&\mbox{ if }\omega\in(0,\pi/2],\\
			D^1&\mbox{ if }\omega \in(\pi/2,\pi).
		\end{cases}
	\]
	where $D$  and $D^1$ are defined in \eqref{def:DiracMIT} and Theorem \ref{thm:sa}\eqref{thm:ncc}, respectively. As defined $D^{sa}$ is self-adjoint. The following two propositions describe basic spectral properties of $D^{sa}$. The first one is about the structure of its essential spectrum and the second one deals with its point spectrum.
	%
	\begin{proposition}\label{prop:spec}
		%
			Let $\omega\in(0,\pi)$. We have
			\[
				\Sp (D^{sa}) = \Sp_{ess} (D^{sa})=
					\begin{cases}
						\R &\mbox{ if }m\leq 0,
						\\
						\R\setminus (-m,m)&\mbox{ if }m\geq0.
					\end{cases}
			\]
	\end{proposition}
	\begin{proposition}\label{prop:pspec}
	Let $\omega\in (0,\pi)$. $D^{sa}$ has no point spectrum in $\R\setminus (-|m|,|m|)$.
	\end{proposition}	
\begin{remark}\label{rem:virial}
		The localization of the point spectrum is a consequence of the Virial identity (see in particular \cite[Section 4.7.2]{Thaller1992} and Section \ref{sec:virial}). Nevertheless, this identity gives no information on the existence of point spectrum in $(-|m|,|m|)$ for negative $m$.
	\end{remark}
	\subsection{Organization of the paper}
	In Section \ref{sec:mainproofsa}, we prove Theorem \ref{thm:sa} and state the main lemmas that we need. Their proofs are gathered in Section \ref{sec:lemproofs}. In Section \ref{sec:physsa}, we discuss the physically relevant self-adjoint extensions. Finally, the spectral properties of Proposition \ref{prop:spec} and Proposition \ref{prop:pspec} are proved in Section \ref{sec:spectrum}.
	\section{Self-adjoint extensions of $D$}\label{sec:mainproofsa}
	%
	%
	%
	%
	In this section, we state the main lemmas on which rely the proofs of Theorem \ref{thm:sa}\eqref{thm:cc}-\eqref{thm:ncc}. Their proofs are detailed in Section \ref{sec:lemproofs}. Note that without loss of generality, we can assume that $m=0$ since $m\sigma_3$ is a bounded self-adjoint operator.
	\subsection{The operator in polar coordinates}
	%
	%
	%
	%
	Let us introduce the polar coordinates in $\Omega_\omega$
	\begin{equation}\label{def:polarCoor}
		\begin{split}
				x(r,\theta) = 
				\begin{pmatrix}
					x_1(r,\theta)\\
					x_2(r,\theta))
				\end{pmatrix} 
				 = 
				\begin{pmatrix}
					r\cos(\theta)\\
					r\sin(\theta)
				\end{pmatrix}= re_{rad}(\theta),
		\end{split}
	\end{equation}
	for $r>0$ and $\theta\in(-\omega,\omega)$, where
	\begin{equation}\label{eq:angBase}
		e_{rad}(\theta) = \begin{pmatrix}
			\cos(\theta)\\
			\sin(\theta)
		\end{pmatrix},
		\quad
		e_{ang}(\theta) = \frac{d}{d\theta}e_{rad}(\theta) = \begin{pmatrix}
			-\sin(\theta)\\
			\cos(\theta)
		\end{pmatrix}.
	\end{equation}
	For further use, we recall the following basic relation
	\begin{equation}\label{eqn:basiceqn}
		i\sigma_3 \sigma\cdot e_{ang} = \sigma\cdot e_{rad}.
	\end{equation}
	For all $\Psi\in L^2(\Omega_\omega, \C^2)$, we get
	\[
		\psi(r,\theta) = \Psi(x(r,\theta))
	\]
	belongs to $L^2((0,+\infty),rdr)\otimes L^2((-\omega,\omega),\C^2)$. 
	In this system of coordinates, the Dirac operator rewrites
	\begin{equation}\label{eq:DiracAng}
		\begin{split}
		D 
		&= -i\sigma\cdot e_{rad}\pa_r - \frac{i\sigma\cdot e_{ang}}{r}\pa_\theta 
		=  -i\sigma\cdot e_{rad}\left(\pa_r +i \frac{\sigma_3}{r}\pa_\theta\right)
		\\
		&=  -i\sigma\cdot e_{rad}\left(\pa_r +\frac{1_2 - K}{2r}\right)
		\end{split}
	\end{equation}
	where
	\begin{equation}\label{op:ang}
		K = \sigma_3\left(
			-2i\pa_\theta
		\right)  + 1_2.
	\end{equation}
	In what follows, we rely on properties of $K$ to build invariant subspaces of $D$.
	\subsection{Study of the operator $K$}
	Remark that for all $r>0$, the boundary matrices write
	\begin{equation}\label{def:matpm}
		\begin{split}
		&
		\mathcal{B}_{\n(re^{i\omega})} = -i\sigma_3 \sigma\cdot e_{ang}(\omega) =: \mathcal{B}_+
		\\
		&
		\mathcal{B}_{\n(re^{-i\omega})} = i\sigma_3 \sigma\cdot e_{ang}(-\omega) =: \mathcal{B}_-
		\end{split}
	\end{equation}
	where $\mathcal{B}_\n$ is defined in \eqref{def:BoundaryCond}.
	Now, let us describe the spectral properties of $K$.
	\begin{lemma}\label{prop:angOp}
		The following holds.
		\begin{enumerate}[(i)]
		\item\label{prop:Ksa}
		The operator $(K,\mathcal{D}(K))$ acting on $L^2((-\omega,\omega),\C^2)$ with $K$ defined in \eqref{op:ang} and
		\[
			\mathcal{D}(K) = \{
				u\in H^1((-\omega,\omega),\C^2)\,: \mathcal{B}_+ u(\omega) = u(\omega) \mbox{ and } \mathcal{B}_- u(-\omega) = u(-\omega)
			\}
		\]
		is self-adjoint and has compact resolvent. 
		\item \label{prop:KSp} Its spectrum is
		 \[
		 	\Sp(K) = \left\{\lambda_\kappa\,, \kappa\in \Z\right\}
		 \]
		 with $\lambda_\kappa := \frac{\pi(1+2\kappa)}{2\omega}$. For $\kappa\in \Z$, we have $\ker\left(
				K-\lambda_\kappa
			\right) = \Span(u_{\kappa})$
		where
		 \[
			u_{\kappa} :=\frac{1}{2\sqrt{\omega}}
				\begin{pmatrix}
				e^{i\theta\frac{\lambda_\kappa-1}{2}}\\
				(-1)^{\kappa+1}ie^{-i\theta\frac{\lambda_\kappa-1}{2}}
			\end{pmatrix},
		 \]
		  and $(u_{\kappa})_{\kappa\in \Z}$ is an orthonormal basis of $L^2((-\omega,\omega),\C^2)$.		 %
		 \item \label{prop:Kcomm}
		We have $\left(\sigma\cdot e_{rad} \right)\mathcal{D}(K)\subset \mathcal{D}(K)$,
		$
			\{K,\sigma\cdot e_{rad}\} = 0 
		$
		and 
		\[
			 u_{-(\kappa+1)} = (-1)^\kappa i (\sigma\cdot e_{rad}) u_{\kappa}.
		\]
		\end{enumerate}
	\end{lemma}
	\begin{remark}
		Thanks to Lemma \ref{prop:angOp}\eqref{prop:Kcomm}, we remark that $\Sp(K)$ is symmetric with respect to $0$. 
	\end{remark}
	\begin{remark}\label{rem:commK}
		The wave functions expansion in angular harmonics for the Dirac operator on $\R^2$ has been a major inspiration for this work. In this case, the operator acting in the angular variable is called the spin-orbit operator and is defined by
		\[
			\begin{split}
			&
			\tilde{K} = -2i\pa_\theta + \sigma_3 = \sigma_3K,
			\\
			&
			\mathcal{D}(\tilde K) = H^1(\R/2\pi\Z,\C^2).
			\end{split}
		\]
		It is self-adjoint and commutes with the Dirac operator on $\R^2$ thus, the eigenspaces of $\tilde{K}$ yield invariant subspaces of the full operator.
		We refer to \cite[Section 4.6]{Thaller1992} where the spherical symmetry in $\R^3$ is extensively studied.
		
		In our case, $\tilde{K}$ does not behave well with respect to the infinite mass boundary conditions. Nevertheless, the slight change we have done overcome this difficulty. Below, we list properties of $K$ that motivates its introduction.
		\begin{enumerate}[(a)]
			\item It is a first order operator in the angular variable $\theta$.
			\item Its domain takes into account the infinite mass boundary conditions and renders $K$ self-adjoint.
			\item It has good anticommutation relations with $D$.
		\end{enumerate}
	\end{remark}
	\subsection{Invariant subspaces of $D$}\label{subsec:stabsubspa}
	Now, we introduce invariant subspaces of $D$ and study the resulting restricted operators. The following lemma is an adaptation of \cite[Theorem XIII.85]{reed-4} to our framework.

	\begin{lemma}\label{lem:ssev}
		We have 
		\[
			L^2((0,+\infty),rdr)\otimes L^2((-\omega,\omega),\C^2) = \oplus_{\kappa\geq0} E_\kappa
		\]
		where $E_\kappa = L^2((0,+\infty),rdr)\otimes \Span(u_{\kappa},u_{-(\kappa+1)})$. Moreover, the following holds.
		\begin{enumerate}[(i)]
			\item \label{lempt:welldefssev} For all $\kappa\in \N$, the operator $(d^\kappa,\mathcal{D}(d^\kappa))$ defined by
			\[
				\begin{split}
					&
					\mathcal{D}(d^\kappa) = \mathcal{D}(D)\cap E_\kappa
					\\
					& d^\kappa = {D}_{\big| { E_\kappa}}
				\end{split}
			\]
			is a well-defined unbounded operator on the Hilbert space $E_\kappa$.
			\item \label{lempt:ssevunitequiv} For all $\kappa\in \N$, the operator $(d^\kappa,\mathcal{D}(d^\kappa))$ is unitarily equivalent to the operator $(\mathsf{ d}_\omega^\kappa,\mathcal{D}(\mathsf{d}_\omega^\kappa))$ defined by
			\[
				\begin{split}
					&
					\mathcal{D}(\mathsf{d}_\omega^\kappa) 
					= 
					\left\{
						\begin{pmatrix}
							a\\b
						\end{pmatrix}
						\in L^2((0,+\infty),\C^2,rdr)\, :\
					\right.
					\\
					&
					\left.
					\quad\quad\quad\quad\quad
						 \int_0^\infty \left(|\dot a|^2 + |\dot{b}|^2 + \frac{|\lambda_\kappa-1|^2}{4r^2}|a|^2+ \frac{|\lambda_{\kappa}+1|^2}{4r^2}|b|^2\right)rdr <+\infty
					\right\},
					\\
					& \mathsf{d}_\omega^\kappa 
					= 
					(-1)^\kappa\left(
						i\sigma_2\left(
								\pa_r + \frac{1}{2r}
							\right)
						+\sigma_1\frac{\lambda_\kappa}{2r}
					\right) = (-1)^{\kappa} \begin{pmatrix} 0 & \partial_r + \frac{\lambda_\kappa + 1}{2r}\\ -\partial_r + \frac{\lambda_\kappa - 1}{2r} & 0\end{pmatrix}.
				\end{split}
			\]
			\item \label{lempt:formequad} Let 
				\[
					v = \sum_{\kappa\in \Z} a_{\kappa}u_\kappa
				\]
				be any element of $\mathcal{D}(D)$, we have
				\[
					\|D v\|^2_{L^2(\Omega_\omega,\C^2)} = \sum_{\kappa\in \Z}\int_0^\infty\left(|\dot{a_\kappa}|^2+|\lambda_\kappa-1|^2\frac{|a_\kappa(r)|^2}{4r^2}\right)rdr = \|\nabla v\|^2_{L^2(\Omega_\omega,\C^2)}.
				\]
				\item \label{lempt:closed}
				For all $\kappa\in \N$, the operators $(D,\mathcal{D}(D))$ and $(\mathsf{ d}_\omega^\kappa,\mathcal{D}(\mathsf{d}_\omega^\kappa))$ are symmetric and closed.
				\item \label{lempt:sa}
				Let 
				$\{(\widetilde{\mathsf{ d}_\omega^\kappa},\mathcal{D}(\widetilde{\mathsf{ d}_\omega^\kappa}))| \ \kappa\geq0\}$ be a family of extensions of the operators $\{(\mathsf{ d}_\omega^\kappa,\mathcal{D}(\mathsf{d}_\omega^\kappa))|\ \kappa \geq0\}$. Denote by $(\widetilde D, \mathcal{D}(\widetilde D))$ the extension of $(D,\mathcal{D}(D))$ which satisfies
				\[
					\mathcal{D}(\widetilde D) = \oplus_{\kappa\geq0}\ \mathcal{D}(\widetilde{\mathsf{ d}_\omega^\kappa}).
				\]
				The operator $(\widetilde D, \mathcal{D}(\widetilde D))$ is self-adjoint if and only if the operators\break $(\widetilde{\mathsf{ d}_\omega^\kappa},\mathcal{D}(\widetilde{\mathsf{ d}_\omega^\kappa}))$ are self-adjoint. In this case, we have
				\[
					Sp(\widetilde D) = \bigcup_{\kappa\in\Z} Sp\big(\widetilde{\mathsf{ d}_\omega^\kappa}\big).
				\]
		\end{enumerate}
	\end{lemma}

	The following lemma concludes our study. Its proof relies on \cite[Theorem VIII.3]{reed-1}, \cite[Theorem X.2]{reed-2} and some properties of the modified Bessel functions (see for instance \cite[Chapter 12]{olver2014asymptotics} or \cite[Chapter 9]{abramowitz1964handbook}).

	\begin{lemma}\label{lem:bessel}
		The following holds.
			\begin{enumerate}[(i)]
				\item Let $\kappa \geq 1$ and $\omega\in(0,\pi)$. The operator $(\mathsf{ d}_\omega^\kappa,\mathcal{D}(\mathsf{d}_\omega^\kappa))$ is self-adjoint. When $\kappa = 0$, $(\mathsf{ d}_\omega^0,\mathcal{D}(\mathsf{d}_\omega^0))$ is self-adjoint as long as $\omega\in(0,\pi/2]$.
				\item For all $\omega\in(\pi/2,\pi),$ $(\mathsf{ d}_\omega^0,\mathcal{D}(\mathsf{d}_\omega^0))$ is not self-adjoint but has infinitely many self-adjoint extensions $(\mathsf{ d}_\omega^{0,\gamma},\mathcal{D}(\mathsf{d}_\omega^{0,\gamma}))$ defined by
				\[
					\begin{split}
						&\mathcal{D}(\mathsf{d}_\omega^{0,\gamma})
						= \mathcal{D}(\mathsf{d}_\omega^{0}) 
						+ \Span (a_++\gamma\sigma_3a_+),
						\\
						&\mathsf{ d}_\omega^{0,\gamma}(a + c_0(a_++\gamma\sigma_3a_+))
						 = \mathsf{ d}_\omega^{0}a + c_0i(a_+-\gamma\sigma_3a_+),
					\end{split}
				\]
			with $a\in\mathcal{D}(\mathsf{d}_\omega^{0}) $, $c_0\in \C$,
			\[
				a_+ : r\mapsto \begin{pmatrix}
				K_{\frac{\lambda_0-1}{2}}(r)
				\\ -i K_{\frac{\lambda_0+1}{2}}(r)
			\end{pmatrix}
			\]
			and $\gamma\in \C$ such that $|\gamma| = 1$.
			\end{enumerate}
	\end{lemma}
	Theorem \ref{thm:sa}\eqref{thm:cc}-\eqref{thm:ncc} follow from Lemmas \ref{prop:angOp}, \ref{lem:ssev} and \ref{lem:bessel}.
\section{Proofs of Lemmas \ref{prop:angOp}, \ref{lem:ssev} and \ref{lem:bessel}}\label{sec:lemproofs}
	In \S \ref{subsec:3.1}, we gather basic results that are necessary in what remains of this section. \S \ref{subsec:3.2}, \S \ref{subsec:3.3} and \S \ref{subsec:3.4} deal with the proofs of Lemmas \ref{prop:angOp}, \ref{lem:ssev} and \ref{lem:bessel}, respectively.
	\subsection{Preliminary study}\label{subsec:3.1}
	The following lemma is about basic spectral properties of the matrices $\mathcal{B}_\vecteur$ defined in \eqref{def:BoundaryCond}.
\begin{lemma}\label{lem:PropB}
		For all unit vector $\vecteur\in \R^2$, the matrix $\mathcal{B}_\vecteur$ satisfies
		\[
			\ker (\mathcal{B}_\vecteur\pm1_2) = \sigma_3\ker (\mathcal{B}_\vecteur\pm1_2)^\perp = \sigma\cdot \vecteur\ker (\mathcal{B}_\vecteur\pm1_2)^\perp.
		\] 
	\end{lemma}
	\begin{proof}
	Since $\{\sigma_3,\mathcal{B}_\vecteur\} = 0$, we have 
		\[
			\sigma_3\ker (\mathcal{B}_\vecteur\pm1_2)^\perp = \sigma_3 \ran (\mathcal{B}_\vecteur\pm1_2) = \ran (\mathcal{B}_\vecteur\mp1_2) = \ker (\mathcal{B}_\vecteur\pm1_2).
		\]
		Moreover, as $\{\sigma\cdot \vecteur,\mathcal{B_\vecteur}\} = 0$, the same proof yields the other equality.
	\end{proof}

	For the sake of completeness, we recall the following standard result on the symmetry of the Dirac operator with infinite mass boundary conditions.
	\begin{lemma}\label{lem:sym}
	 	The operator $(D, \mathcal{D}(D))$ is symmetric and densely defined.
	\end{lemma}
	\begin{proof}
		Let $u,v\in\mathcal{D}(D)$. Since $\Omega_\omega$ is a Lipschitz domain, an integrations by parts yields
		\[
			\braket{D u,v}_{L^2(\Omega_\omega,\C^2)} - \braket{u,D v}_{L^2(\Omega_\omega,\C^2)}
			= \braket{-i\sigma\cdot \n u,v}_{L^2(\partial\Omega_\omega,\C^2)},
		\]
		(see  \cite[Section 3.1.2]{necas2011direct}).
		Thanks to Lemma \ref{lem:PropB}, almost everywhere on the boundary we have
		\[
			-i(\sigma\cdot \n) u\in \sigma\cdot \n\ker (\mathcal{B}_\n-1_2) = \ker (\mathcal{B}_\n-1_2)^\perp.
		\]
		Thus $\braket{-i\sigma\cdot \n u,v}_{L^2(\partial\Omega_\omega,\C^2)}  = 0$ and we obtain
		 \[
		 	\braket{D u,v}_{L^2(\Omega_\omega,\C^2)} = \braket{u,D v}_{L^2(\Omega_\omega,\C^2)}.
		 \]
	\end{proof}
	\subsection{Study of the angular part :  proof of Lemma \ref{prop:angOp}}\label{subsec:3.2}
		The proof is divided into several steps.
		\subsubsection*{Step 1: symmetry of $K$.}
		Let $u,v\in \mathcal{D}(K)$, an integration by parts and Lemma \ref{lem:PropB} yield
		\begin{align}\label{eq:intByPartsK}
				%
				&\braket{Ku,v}_{L^2((-\omega,\omega),\C^2)}-\braket{u,Kv}_{L^2((-\omega,\omega),\C^2)} \nonumber\\
				&=
				\int_{-\omega}^\omega\pa_\theta\big(\braket{-2i\sigma_3u,v}_{\C^2}\big)d\theta 
				\nonumber\\
				&=\braket{-2i\sigma_3u(\omega),v(\omega)}_{\C^2}-\braket{-2i\sigma_3u(-\omega),v(-\omega)}_{\C^2} =0.
				%
		\end{align}
		Hence, $K$ is symmetric. 
		\subsubsection*{Step 2: self-adjointness of $K$}
		Let $u\in \mathcal{D}(K^*)$. Using test functions in\break $C^\infty_c((-\omega,\omega),\C^2)\subset \mathcal{D}(K)$, we remark that the distribution $Ku$ belongs to $L^2((-\omega,\omega),\C^2)$ which gives $u\in H^1((-\omega,\omega),\C^2)$. Performing again integration by parts \eqref{eq:intByPartsK}, we get
		\[
			\begin{split}
				&\sigma_3u(\omega) \in \ker (\mathcal{B}_+-1_2)^\perp,\\
				&\sigma_3u(-\omega) \in \ker (\mathcal{B}_--1_2)^\perp.
			\end{split}
		\]
		Thanks to Lemma \ref{lem:PropB}, we obtain that $u$ belongs to $\mathcal{D}(K)$ and thus $K$ is self-adjoint. Finally, the compact Sobolev embedding 
		\[
			H^1((-\omega,\omega),\C^2)\hookrightarrow L^2((-\omega,\omega),\C^2),
		\]
		implies that $K$ has compact resolvent. Hence, its spectrum is discrete and it concludes the proof of Lemma \ref{prop:angOp}\eqref{prop:Ksa}.
		\subsubsection*{Step 3: study of $Sp(K)$}
		 Let $\lambda\in \R$, we look for solutions of 
		\begin{equation}\label{eq:Kspec}
			Ku = \lambda u
		\end{equation}
		belonging to $\mathcal{D}(K)$. Remark that without taking the boundary conditions into account, the set of solutions of \eqref{eq:Kspec} is the vector space
		\[
			E_\lambda^1  := \Span \left(
				\begin{pmatrix}
					e^{i\theta\frac{\lambda-1}{2}}\\
					0
				\end{pmatrix},
				\begin{pmatrix}
					0\\
					e^{-i\theta\frac{\lambda-1}{2}}
				\end{pmatrix}
			\right).
		\]
		Assume $u \in E^1_\lambda\cap \mathcal{D}(K)$. In particular, $u$ writes
		\[
			u =\begin{pmatrix}
				ae^{i\theta\frac{\lambda-1}{2}}\\
				be^{-i\theta\frac{\lambda-1}{2}}
			\end{pmatrix}
		\]
		for some constants $a,b\in\C$. Using \eqref{eqn:basiceqn}, the boundary conditions read
		 {\small
		 \[
		 	\begin{split}
			&
		 	u(\omega) = \mathcal{B}_+ u(\omega) = -\sigma\cdot e_{rad}(\omega)u(\omega) = -\begin{pmatrix}
				0&e^{-i\omega}
				\\
				e^{i\omega}&0
			\end{pmatrix}u(\omega)
			= 
			\begin{pmatrix}
				-be^{-i\omega\frac{\lambda+1}{2}}\\
				-ae^{i\omega\frac{\lambda+1}{2}}
			\end{pmatrix},
			\\
			&
			u(-\omega) = \mathcal{B}_- u(-\omega) = \sigma\cdot e_{rad}(-\omega)u(-\omega) = \begin{pmatrix}
				0&e^{i\omega}
				\\
				e^{-i\omega}&0
			\end{pmatrix}u(-\omega)
			= 
			\begin{pmatrix}
				be^{i\omega\frac{\lambda+1}{2}}\\
				ae^{-i\omega\frac{\lambda+1}{2}}
			\end{pmatrix}.
			\end{split}
		 \]
		 }

		\noindent It yields
		 \[
		 	a = be^{i\omega\lambda} = -be^{-i\omega\lambda}.
		 \]
		 Hence, there is a nontrivial solution of \eqref{eq:Kspec} belonging to $\mathcal{D}(K)$ if and only if $e^{2i\omega\lambda} = -1$. We deduce that the spectrum of $K$ is 
		 \[
		 	\Sp(K) = \left\{\frac{\pi(1+2\kappa)}{2\omega}\,,\ \kappa\in \Z\right\}.
		 \]
		 For $\kappa \in \Z$, define $\lambda_\kappa := \frac{\pi(1+2\kappa)}{2\omega}$. If $\kappa$ is even, we have $a = ib$ and
		 \[
		 	\ker\left(
				K-\lambda_\kappa
			\right)
			=
			\Span
				\begin{pmatrix}
				e^{i\theta\frac{\lambda_\kappa-1}{2}}\\
				-ie^{-i\theta\frac{\lambda_\kappa-1}{2}}
			\end{pmatrix}
			,
		 \]
		 if $\kappa$ is odd, we have $a = -ib$ and
		  \[
		 	\ker\left(
				K-\lambda_\kappa
			\right)
			=
			\Span
				\begin{pmatrix}
				e^{i\theta\frac{\lambda_\kappa-1}{2}}\\
				ie^{-i\theta\frac{\lambda_\kappa-1}{2}}
			\end{pmatrix}
			.
		 \]
		This proves Lemma \ref{prop:angOp}\eqref{prop:KSp}.
		\subsubsection*{Step 4: the commutation relation}
		Since $\sigma\cdot e_{rad}$ commutes with $\mathcal{B}_\n$, we obtain
		\[
			(\sigma\cdot e_{rad})\mathcal{D}(K)\subset \mathcal{D}(K).
		\]
		We also have
		\[
			\begin{split}
				K\sigma\cdot e_{rad} &= \sigma_3\left(\sigma\cdot e_{rad}(-2i\pa_\theta) -2i\sigma\cdot e_{ang}\right) + \sigma\cdot e_{rad}
				\\
				& 
				= -\sigma\cdot e_{rad}\sigma_3(-2i\pa_\theta)- \sigma\cdot e_{rad} = -\sigma\cdot e_{rad}K.
			\end{split}
		\]
		This ends the proof of Lemma \ref{prop:angOp}\eqref{prop:Kcomm}.
		%
		
		%
		%
		%
		\subsection{Invariant subspaces: proof of Lemma \ref{lem:ssev}}\label{subsec:3.3}
			Let us remark that the direct sum decomposition is a direct consequence of Lemma \ref{prop:angOp}\eqref{prop:KSp}.
			What remains of the proof is divided into several steps.
			\subsubsection*{Proof of Points \eqref{lempt:welldefssev} and \eqref{lempt:ssevunitequiv}}
			These points follow from identity \eqref{eq:DiracAng}
			\[
				D 
				=  -i\sigma\cdot e_{rad}\left(\pa_r +\frac{1_2 - K}{2r}\right)
			\]
			and Lemma \ref{prop:angOp}\eqref{prop:Kcomm}. 
			Indeed, for all $\kappa\in \N$ and all $v\in E_\kappa$ there exist $a,b\in L^2((0,+\infty),rdr)$ such that for all $r>0$ and all $\theta\in(-\omega,\omega)$, $v$ writes
			\[
				v(r,\theta) =  a(r)u_\kappa(\theta) +b(r)u_{-(\kappa+1)}(\theta).
			\]
			If $v\in H^1(\Omega_\omega,\C^2)$, since $-i\sigma_3\pa_\theta = \frac{K-1}{2}$, we have
					\[
						\|\nabla v\|_{L^2(\Omega_\omega,\C^2)}^2 = \int_0^\infty \left(|\dot a|^2 + |\dot{b}|^2 + \frac{|\lambda_\kappa-1|^2}{4r^2}|a|^2+ \frac{|\lambda_{-(\kappa+1)}-1|^2}{4r^2}|b|^2\right)rdr
					\]
			and
			\[
					D v = d^\kappa v =  (-1)^{\kappa+1}u_{-(\kappa+1)}\left(\dot{a} + \frac{1-\lambda_\kappa}{2r}a\right) + (-1)^{\kappa} u_\kappa\left(\dot{b} + \frac{1+\lambda_\kappa}{2r}b\right).
			\]
			This ends this part of the proof.
			\subsubsection*{Proof of Points \eqref{lempt:formequad} and \eqref{lempt:closed}}
			 Let $v\in\mathcal{D}(D)$. Decomposing $v$ in the orthonormal basis $(u_\kappa)_{\kappa\in\Z}$, it writes
				\[
					v = \sum_{\kappa\in \Z} a_{\kappa}u_\kappa.
				\]
				Using Lemma \ref{lem:radfun}, we have

					\begin{align*}
					\|D v\|^2_{L^2(\Omega_\omega,\C^2)} 
					&= 
					\sum_{\kappa\in \Z}\int_0^\infty\left|\dot{a_\kappa} + \frac{1-\lambda_\kappa}{2r}a_\kappa\right|^2rdr\\
					&=
					\sum_{\kappa\in \Z}\int_0^\infty\left(|\dot{a_\kappa}|^2+|\lambda_\kappa-1|^2\frac{|a_\kappa(r)|^2}{4r^2}\right)rdr
					\\
					& = \|\nabla v\|^2_{L^2(\Omega_\omega,\C^2)}.
					\end{align*}
			and Lemma \ref{lem:ssev}\eqref{lempt:closed} is proved.
			\subsubsection*{Proof of Point \eqref{lempt:sa}}
			This last point is proved as in \cite[Lemma 4.15]{Thaller1992} using \cite[Theorem VIII.3]{reed-1}.
			%
		%
		%
		\subsection{Proof of Lemma \ref{lem:bessel}}\label{subsec:3.4}
		%
		%
		%
		%
		Let $\kappa\in \N$.
		In this proof, we apply the basic criterion for self-adjointness \cite[Theorem VIII.3]{reed-1}. In particular, we have to study the vector spaces
		\[
			\ker(({\mathsf{d}_\omega^\kappa})^*\pm i1_2).
		\]
		Remark that
		\[
			\mathcal{D}(({\mathsf{d}_\omega^\kappa})^*) 
			\subset \left\{ \begin{pmatrix}
				a\\b
			\end{pmatrix}\in L^2((0,\infty),rdr)^2\,:\
				{\mathsf{d}_\omega^\kappa}\begin{pmatrix}
				a\\b
			\end{pmatrix}\in L^2((0,\infty),rdr)^2
			 \right\}.
		\]
		Since $\{\mathsf{d}_\omega^\kappa,\sigma_3\}=0$, we obtain
		\[
			\ker(({\mathsf{d}_\omega^\kappa})^*-i1_2) = \sigma_3\ker(({\mathsf{d}_\omega^\kappa})^*+i1_2).
		\]
		Hence, it remains to look for $L^2((0,\infty),rdr)^2$-solutions of 
		\begin{equation}\label{eqn:soladj}
			\left(\mathsf{d}_\omega^\kappa- i1_2\right) \begin{pmatrix}
				a\\b
			\end{pmatrix} = 0.
		\end{equation}
		It is well known that the set of solutions is a vector space of dimension $2$ and moreover, the solutions are smooth on $(0,\infty)$. Remark that
		\[
			\left(\mathsf{d}_\omega^\kappa+ i1_2\right)\left(\mathsf{d}_\omega^\kappa- i1_2\right) = (\mathsf{d}_\omega^\kappa)^2+1_2,
		\] 
		which implies
		{\small
		\begin{multline*}
		 	0 =\left((\mathsf{d}_\omega^\kappa)^2+ 1_2\right) \begin{pmatrix}
				a\\b
			\end{pmatrix} 
			\\= -\frac{1}{r^2}\begin{pmatrix}
				r^2\pa_r^2 + r\pa_r - \left(r^2+\frac{(\lambda_\kappa-1)^2}{4}\right) & 0\\
				0 & r^2\pa_r^2 + r\pa_r - \left(r^2+\frac{(\lambda_\kappa+1)^2}{4}\right)
			\end{pmatrix} \begin{pmatrix}
				a\\b
			\end{pmatrix}.
		\end{multline*}
		}
		
	\noindent	Thus, $a$ and $b$ are modified Bessel functions (see \cite[Chapter 12, Section 1]{olver2014asymptotics}  and \cite{abramowitz1964handbook}) of parameters $\frac{\lambda_\kappa-1}{2}$ and $\frac{\lambda_\kappa+1}{2}$, respectively. 
		The modified Bessel functions of the first kind do not belong to $L^2((1,\infty),rdr)$. Consequently, for $(a,b)$ to belong to $L^2((1,\infty),rdr)^2$ $a$ and $b$ necessarily write $a=a_0K_{\frac{\lambda_\kappa-1}{2}}$ and $b = b_0K_{\frac{\lambda_\kappa+1}{2}}$ with $a_0,b_0\in \C$. Recall that $K_{\nu}$ denotes the modified Bessel function of the second kind of parameter $\nu\in \R$. It is known that $K_{\nu} = K_{-\nu}$ and
		\begin{equation}\label{eq:difBess}
			\begin{split}
				\dot{K_{|\nu|}}(r) + \frac{|\nu|}{r}K_{|\nu|}(r) = -K_{|\nu|-1}(r)
			\end{split}
		\end{equation}
		for all $r>0$.
		Thanks to Remark \ref{rem:besselasym}, for $b$ to belong to $L^2((0,1),rdr)$ one needs
		$
			\lambda_\kappa<1.		
		$
		We have
			\begin{enumerate}[(a)]
				\item\label{cas1} $\lambda_\kappa\geq 3/2 \mbox{ for any }\kappa \geq 1 \mbox{ and any }\omega\in(0,\pi),$
				\item\label{cas2}$\lambda_0\geq1 \mbox{ for any }\omega\in(0,\pi/2],$
				\item\label{cas3}$\lambda_0<1 \mbox{ for any }\omega\in(\pi/2,\pi).$
			\end{enumerate}
		Hence,  in Cases \eqref{cas1} and \eqref{cas2}, we get $b_0=0$. Taking into account \eqref{eqn:soladj} and \eqref{eq:difBess} we also get $a_0 = 0$ which implies
		\[
			\ker(({\mathsf{d}_\omega^\kappa})^*\pm i1_2) = \{0\}
		\]
		and \cite[Theorem VIII.3]{reed-1} ensures that $(\mathsf{d}_\omega^\kappa,\mathcal{D}(\mathsf{d}_\omega^\kappa))$ is a self-adjoint operator. 
		In Case \eqref{cas3}, we get
		\[
			 \begin{pmatrix}
				a\\b
			\end{pmatrix} 
			\in \Span (a_+^0),\mbox{ with }a_+^0 =  
			 \begin{pmatrix}
				K_{\frac{\lambda_0-1}{2}}(r)
				\\- i K_{\frac{\lambda_0+1}{2}}(r)
			\end{pmatrix} .
		\]
		Actually, $a_+^0$ belongs to $\mathcal{D}(({\mathsf{d}_\omega^\kappa})^*)$ which yields
		\[
			\ker(({\mathsf{d}_\omega^\kappa})^*- i1_2) = \Span (a_+^0) \mbox{ and } \ker(({\mathsf{d}_\omega^\kappa})^*+ i1_2) = \Span (\sigma_3 a_+^0).
		\]
		We conclude thanks to \cite[Theorem X.2]{reed-2}.
		\section{Distinguished self-adjoint extensions of $D$}\label{sec:physsa}
		The goal of this section is to prove Propositions \ref{prop:charge_conj} and \ref{prop:dist_ext} about the distinguished extensions of $\big(D,\mathcal{D}(D)\big)$ when $\omega\in(\pi/2,\pi)$.
		\subsection{Proof of Proposition \ref{prop:charge_conj}}
		The anticommutation of $C$ with $D$ is straightforward. The only thing left to prove is the following lemma.
		\begin{lemma}\label{lemma:}Let $\omega\in(\pi/2,\pi)$ and let $\gamma\in\C$ be such that $|\gamma| = 1$. The following statements are equivalent.
			\begin{enumerate}[(a)]
				\item\label{itm:a} $\gamma = \pm 1$.
				\item\label{itm:b} $C \mathcal{D}(D^\gamma) \subset \mathcal{D}(D^\gamma)$ and $D^\gamma C = -C D^\gamma$.
			\end{enumerate}
		\end{lemma}
		\begin{proof} Let $u\in \mathcal{D}(D^\gamma)$. Thanks to Theorem \ref{thm:sa}\eqref{thm:ncc}, there exist $v\in \mathcal{D}(D)$ and $c_0\in\C$ such that $u = v + c_0(v_+ + \gamma v_-)$. The following equalities hold:
		\[
			C v_+ = i v_+,\quad C v_- = i v_-.
		\]
		As $C u = C v + i\overline{c_0}(v_+ + \overline{\gamma}v_-)$ we have $Cu \in \mathcal{D}(D^\gamma)$ if and only if $\gamma\in\mathbb{R}$, that is $\gamma = \pm 1$. Now, let $\gamma=\pm1$, we have
		\[
			D^\gamma Cu = D Cv - \overline{c_0}(v_+ - \gamma v_-) = -C D v - \overline{c_0}(v_+ - \gamma v_-).
		\]
		As $D^\gamma(v_+ + \gamma v_-) = i(v_+ -\gamma v_-)$, we get $CD^\gamma(v_+ + \gamma v_-) = (v_+ - \gamma v_-)$ which yields $D^\gamma Cu = -CD^\gamma u$.
		\end{proof}
		\subsection{Proof of Proposition \ref{prop:scaling}}
		Proposition \ref{prop:scaling} is a consequence of the following lemma.
		\begin{lemma}\label{lem:scal}
			Let  $\alpha>0$ and $\gamma=e^{is}\in \C$  for $s\in [0,2\pi)$. The unitary map
			\begin{equation}\label{eqn:transunitmap}
				\begin{array}{lll}
					\mathcal{V}_\alpha : &L^2(\Omega_\omega,\C^2)&\longrightarrow L^2(\Omega_\omega,\C^2)
					\\
					&u&\longmapsto \alpha u(\alpha\cdot)
				\end{array}
			\end{equation}
			satisfies 
			\[
				\begin{split}
					&\mathcal{V}_\alpha^{-1}(-i\sigma\cdot \nabla)\mathcal{V}_\alpha  = \alpha (-i\sigma\cdot \nabla),
					\\
					&\mathcal{V}_\alpha\mathcal{D}(D^\gamma) = 
					\begin{cases}
						\mathcal{D}(D^\gamma) & \mbox{ if }\gamma = \pm 1 (\text{ i.e. if } s\in\{0,\pi\}),\\
						\mathcal{D}(D^{\tilde \gamma}) & \mbox{ otherwise},
					\end{cases}
				\end{split}
			\]
			where $\tilde \gamma = e^{2i\arctan\left(\frac{\tan(s/2)}{\alpha^{\lambda_0}}\right)}$.
		\end{lemma}
		\begin{proof}
			Let $\alpha>0$ and $\gamma = e^{is}\in \C$. As $\mathcal{D}(\mathsf{d}_\omega^\kappa)$ is scale-invariant for all $\kappa\geq 1$, using Lemma \ref{lem:ssev} we are reduced to investigate the operator $\mathsf{d}_\omega^{0,\gamma}$. 
			Thanks to Remark \ref{rem:besselasym} we have 
			\[
				\mathcal{V}_\alpha\mathcal{D}(D^{\pm 1}) = 
						\mathcal{D}(D^{\pm 1})
			\]
			and for $\gamma\ne -1$ and $r>0$,
			\[
				(1+\gamma)\begin{pmatrix}
				K_{\frac{1-\lambda_0}{2}}(\alpha r)
				\\ -i\frac{1-\gamma}{1+\gamma} K_{\frac{1+\lambda_0}{2}}(\alpha r)
			\end{pmatrix}
			\sim_{r\rightarrow 0}
			\frac{1+\gamma}{\alpha^{\frac{1-\lambda_0}{2}}}
			\begin{pmatrix}
				K_{\frac{1-\lambda_0}{2}}( r)
				\\ -i\alpha^{-\lambda_0}\frac{1-\gamma}{1+\gamma} K_{\frac{1+\lambda_0}{2}}( r)
			\end{pmatrix}
			.
			\]
			We have
			\[
				-i\frac{1-\gamma}{1+\gamma}
				=-\tan(s/2)
			\]
			which rewrites
			\[
				-i\alpha^{-\lambda_0}\frac{1-\gamma}{1+\gamma} = - \alpha^{-\lambda_0}\tan(s/2)=-i\frac{1-\tilde \gamma}{1+\tilde \gamma}
			\]
			for $\tilde \gamma = e^{2i\arctan\left(\frac{\tan(s/2)}{\alpha^{\lambda_0}}\right)}$. This ensures that $\mathcal{V}_\alpha\mathcal{D}(D^\gamma) = \mathcal{D}(D^{\tilde \gamma})$ and the result follows.
		\end{proof}
		\subsection{Proof of Proposition \ref{prop:dist_ext}}
		To prove Proposition \ref{prop:dist_ext} it is enough to prove the following lemma.
		\begin{lemma}\label{lem:sobinje} Let $\omega\in(\pi/2,\pi)$ and let $\nu_0$ be as defined in Theorem \ref{thm:sa}. The following holds true.
			\begin{enumerate}[(i)]
				\item\label{spaceok} The function 
					\[(r\cos(\theta),r\sin(\theta))\in\Omega_\omega \mapsto K_{\nu_0}(r)u_0(\theta)\]
					 belongs to $H^{1/2}(\Omega_\omega)$.
				\item\label{spacenotok} The function
					\[(r\cos(\theta),r\sin(\theta))\in\Omega_\omega \mapsto K_{\nu_0+1}(r)u_{-1}(\theta)\]
					does not belong to $H^{1/2}(\Omega_\omega)$.
			\end{enumerate}
		\end{lemma}
		\begin{proof} 
		Using \cite[Cor. 4.53.]{MR2440993}, we have $H^{1/2}(\Omega_\omega)\hookrightarrow L^4(\Omega_\omega)$ and thanks to Remark \ref{rem:besselasym} we get
		\[
			\left|
				K_{\nu_0+1}(r)u_{-1}(\theta)
			\right|^4r= \frac{1}{2^3\omega^2}
			|K_{\nu_0+1}(r)|^4r \sim_{r\to0} \frac{C}{r^{4(\nu_0+1)-1}}
		\]
		for all $r>0$ and $\theta\in(-\omega,\omega)$.
		Since
		\[
			\nu_0 = \frac{\pi-2\omega}{4\omega}>-1/2,
		\]
		 this function does not belong to $L^4(\Omega_\omega)$ 
		and Lemma \ref{lem:sobinje}\eqref{spacenotok} is proved.
		Let us prove Lemma \ref{lem:sobinje}\eqref{spaceok}. Let $r>0$ and $\theta\in(-\omega,\omega)$, we have
		\[
			|\nabla K_{\nu_0}u_0|^2(r,\theta) = \frac{1}{4\omega}|\pa_r K_{\nu_0}(r)|^2 + \frac{|K_{\nu_0}(r)|^2}{4\omega r^2}2|\nu_0|^2.
		\]
		Thanks to \eqref{eq:difBess} and Remark \ref{rem:besselasym}, $K_{\nu_0}u_0$ belongs to $W^{1,p}(\Omega_\omega)$ as soon as 
		\[
			1\leq p<\frac{2}{|\nu_0|+1}.
		\]
		Since we have
		\[
			\min_{\omega\in(\pi/2,\pi)}\,\frac{2}{|\nu_0|+1} = \frac{8}{5}>\frac{4}{3}
		\]
		and $W^{1,4/3}(\Omega_\omega)\hookrightarrow H^{1/2}(\Omega_\omega)$, we get Lemma \ref{lem:sobinje}\eqref{spaceok}.
		\end{proof}
		\section{Spectrum of $D^{sa}$}\label{sec:spectrum}
		This section is devoted to the proofs of Propositions \ref{prop:spec} and \ref{prop:pspec}.
		\subsection{Proof of Proposition \ref{prop:spec}}
		The proof is divided into three steps. In Steps 1 and 2 we construct Weyl sequences for the Dirac operator with infinite mass boundary conditions on $\Omega_{\pi/2}$ and the free Dirac operator in $\mathbb{R}^2$ denoted by $D_0$, respectively. For a general $\omega\in(0,\pi)$, Weyl sequences for the Dirac operator with infinite mass boundary conditions in $\Omega_{\omega}$ can be obained using adequate cut-off functions. The last step ensures that the Weyl sequences actually yield the whole essential spectrum.

		\emph{Step 1: Weyl sequences for $m\geq 0$ on $\R^2$.}
		Let $\chi : [0,+\infty)\mapsto [0,1]$ be a $C^\infty$-smooth function such that
		\begin{equation}\label{eqn:defWeylchi}
			\chi(x) = \begin{cases}
				1&\mbox{ if }x<1\\
				0&\mbox{ if }x>2.
			\end{cases}
		\end{equation}
		Let $\lambda>m$ and define
		 \[
			u:\left\{\begin{array}{lcl}
			\R^2&\longrightarrow&\C^2\\
			(x_1,x_2)&\longmapsto& \begin{pmatrix}
				\sqrt{\frac{\lambda+m}{\lambda-m}}\\1
			\end{pmatrix}
			e^{ix_1\sqrt{\lambda^2-m^2}}.
			\end{array}\right.
		\]
		In particular, remark that
		$
			(-i\sigma\cdot\nabla+m\sigma_3-\lambda)u = 0.
		$
		Now, for $n>0$, we define the sequence of functions
		\[
		u_n(x_1,x_2) = u(x_1,u_2)\chi(|x_1|/n)\chi(|x_2|/n) \in H^1(\R^2,\C^2).
		\]
		We get
		\[
			\|u_n\|_{L^2(\R^2,\C^2)}^2 = 2 \Big(\frac{\lambda}{\lambda - m}\Big) n^2 \|\chi\|_{L^2(\R)}^4
		\]
		and
		\begin{align*}
			&\|(-i\sigma\cdot\nabla +m\sigma_3 - \lambda)u_n\|_{L^2(\R^2,\C^2)}^2\\
			 &\leq \frac{4}{n^2} \Big(\frac{\lambda}{\lambda - m}\Big)\int_{\R^2} \big(\chi'(|x_1|/n)\chi(|x_2|/n)\big)^2 dx_1dx_2
			\\ &= 4\Big(\frac{\lambda}{\lambda - m}\Big)\int_{\R^2} \big(\chi'(|x_1|)\chi(|x_2|)\big)^2 dx_1dx_2.
		\end{align*}
		Thus, we obtain
		\[
			\frac{\|(-i\sigma\cdot\nabla +m\sigma_3 - \lambda)u_n\|_{L^2(\R^2,\C^2)}^2}{\|u_n\|_{L^2(\R^2,\C^2)}^2} \rightarrow 0,\quad\text{when } n\rightarrow+\infty.
		\]
		In particular, $\lambda$ belongs to the spectrum of the free Dirac operator and thus $(m,+\infty)\subset Sp(D_0)$. As the spectrum of a self-adjoint operator is closed, the end-point $m$ also belongs to the spectrum. Recall that $C$ is the charge conjugation operator introduced in \eqref{eqn:def_chargeconj}. The same reasoning yield that the sequence $(Cu_n)_{n>0}$ is also a Weyl sequence but for the value $-\lambda$. In particular we obtain
		\[
			(-\infty,-m]\cup[m,+\infty) \subset Sp(D_0).
		\]
		As the set on the left-hand side is not discrete, actually we have proved that
		\[
			(-\infty,-m]\cup[m,+\infty) \subset Sp_{ess}(D_0).
		\]
		Finally, localizing the sequences $(u_n)_{n>0}$ and $(Cu_n)_{n>0}$ inside $\Omega_\omega$ and away from the boundary (with well chosen cut-off functions), we obtain
		\[
					(-\infty,-m]\cup[m,+\infty) \subset Sp_{ess}(D^{sa}).
		\]

		\emph{Step 2: Weyl sequences for $m<0$ on $\Omega_{\pi/2}$.} Let $\lambda\in \R$ and define
		\[
			u:\left\{\begin{array}{lcl}
			\Omega_{\pi/2}&\longrightarrow&\C^2\\
			(x_1,x_2)&\longmapsto&\begin{pmatrix}
				1\\-i
			\end{pmatrix}
			e^{mx_1-i\lambda x_2}.
			\end{array}\right.
		\]
		In particular, remark that we have $(-i\sigma\cdot\nabla +m\sigma_3 - \lambda)u = 0$ and $\mathcal{B}_{-e_1}u = u$ on $\partial \Omega_{\pi/2}$ where $e_1=(1,0)^T$. We define the sequence of functions $(u_n)_{n>0}$ by
		\[
			u_n(x_1,x_2) = u(x_1,x_2)\chi(|x_2|/n),\quad \text{ for } (x_1,x_2) \in \Omega_{\pi/2}
		\]
		and with $\chi$ defined in \eqref{eqn:defWeylchi}. Note that as constructed, $u_n \in \mathcal{D}(D^{sa})$. We get
		\begin{align*}
			\|u_n\|^2_{L^2(\Omega_{\pi/2},\C^2)} &= 2\left(\int_0^{+\infty}e^{2mx_1}dx_1\right)\left(\int_\R |\chi(|x_2|/n)|^2dx_2\right)\\
			& = \frac{n}{|m|}\int_\R |\chi(|x_2|)|^2dx_2
		\end{align*}
		and
		\begin{align*}
			\|(D^{sa}-\lambda)u_n\|^2_{L^2(\Omega_{\pi/2},\C^2)} 
			&= \frac{2}{n^2}\||\pa_2 \chi(|x_2|/n)|e^{mx_1}\|^2_{L^2(\Omega_{\pi/2},\C^2)} 		\\
			&	=  \frac{2}{|m|n}\int_\R |\pa_2\chi(|x_2|)|^2dx_2.
		\end{align*}
		This proves that $(u_n)_{n>0}$ is a Weyl sequence and $\lambda$ belongs to the spectrum of $D^{sa}$. We get
$\R \subset Sp(D^{sa})$ which actually read
		\[
			Sp(D^{sa}) = Sp_{ess}(D^{sa}) = \R.
		\]
This proof can be adapted to the domain $\Omega_\omega$ with $\omega\in(0,\pi)$ and negative mass using Remark \ref{rem:turn} and adequate cut-off functions.

		\emph{Step 3: Reverse inclusion for the essential spectrum.} The only thing left to investigate is the case $m>0$.
		 Thanks to Remark \ref{rem:regudom}, we have $\mathcal{D}(D^{sa})\subset H^{3/4-\eps}(\Omega_\omega,\C^2)$ for all $\eps\in(0,1/4)$. Hence, for all $u\in\mathcal{D}(D^{sa})$, an integration by parts yields
		 \[
		 	\begin{split}
				2\Re\braket{-i\sigma\cdot \nabla u,\sigma_3u}_{L^2(\Omega_\omega,\C^2)}
				%
				%
				%
				= \|u\|^2_{L^2(\pa\Omega_\omega,\C^2)}.
			\end{split}
		 \]
		 It gives
		 \[
		 	\begin{split}
		 	\|D^{sa}u\|^2_{L^2(\Omega_\omega,\C^2)} 
			&
			= \|\sigma\cdot \nabla u\|^2_{L^2(\Omega_\omega,\C^2)} + m^2\|u\|^2_{L^2(\Omega_\omega,\C^2)} + m\|u\|^2_{L^2(\pa\Omega_\omega,\C^2)}.
			\end{split}
		 \]
		 This ensures that whenever $m\geq 0$, the spectrum of $(D^{sa})^2$ is included in $[m^2,+\infty)$ and in particular, we have
		 \[
		 	\Sp(D^{sa})\subset (-\infty,-m]\cup[m,+\infty).
		 \]
		It concludes the proof of Proposition \ref{prop:spec}.
		 \subsection{Proof of Proposition \ref{prop:pspec}}\label{sec:virial}
		Assume that $\lambda \in Sp(D^{sa})$ is an eigenvalue. Let $v$ be a normalized eigenfunction associated with $\lambda$. For $\alpha>0$, recall that $\mathcal{V}_\alpha$ is the unitary map introduced in \eqref{eqn:transunitmap}.
Thanks to Proposition \ref{prop:scaling}, $\mathcal{V}_\alpha\big(\mathcal{D}(D^{sa})\big) \subset \mathcal{D}(D^{sa})$. Moreover, we have
		\[\begin{split}
			0 
			&= \braket{(D^{sa} - \lambda)v, \mathcal{V}_\alpha v}_{L^2(\Omega_\omega,\C^2)} 
			= \braket{v,(D^{sa} - \lambda)\mathcal{V}_\alpha v}_{L^2(\Omega_\omega,\C^2)}
			\\
			&= \braket{v,\mathcal{V}_\alpha(\alpha (D^{sa}-\lambda) +(\alpha-1)(\lambda-m\sigma_3))v}_{L^2(\Omega_\omega,\C^2)}
			\\
			&= (\alpha-1)\braket{v,\mathcal{V}_\alpha(\lambda-m\sigma_3)v}_{L^2(\Omega_\omega,\C^2)}.
		\end{split}\]
		For $\alpha\neq 1$ we obtain
		\[
			0 = \braket{(\lambda 1_2 -m\sigma_3)v, \mathcal{V}_\alpha v}_{L^2(\Omega_\omega,\C^2)}.
		\]
		Since $\Big(\alpha\in(0,+\infty) \mapsto \mathcal{V}_\alpha v \in L^2(\Omega_\omega,\C^2)\Big)$ is continuous, taking the limit $\alpha\rightarrow 1$ gives
		\begin{equation}\label{eqn:vir1}
			\lambda\|v\|^2_{L^2(\Omega_\omega,\C^2)} = m\braket{\sigma_3v, v}_{L^2(\Omega_\omega,\C^2)}
		\end{equation}
		which yields
		\[
			|\lambda|\|v\|^2_{L^2(\Omega_\omega,\C^2)} = |m| \big|\braket{\sigma_3v, v}_{L^2(\Omega_\omega,\C^2)}\big|\leq |m|\|v\|^2_{L^2(\Omega_\omega,\C^2)}
		\]
		and necessarily $|\lambda|\leq |m|$. Assume by contradiction that $\lambda = m$. Writing $v = (v_1,v_2)$, \eqref{eqn:vir1} implies that $v_2$ vanishes on $\Omega_\omega$. Taking the infinite mass boundary conditions into account, we get $v_1|_{\partial\Omega_\omega} = 0$ which implies that $v_1 \in H_0^1(\Omega_\omega,\C)$. Moreover, since $(v_1, 0)$ is an eigenfunction of $D^{sa}$, $v_1$ satisfies the Cauchy-Riemann equation
		\[
			-i(\pa_1+i\pa_2)v_1 = 0 \mbox{ in }\Omega_\omega.
		\]
		An integration by parts gives $\|\nabla v_1\|_{L^2(\Omega_\omega, \C^2)} = 0$ and $v_1$ has to vanish. This gives us the wanted contradiction and the case $\lambda = -m$ can be done similarly. It ends the proof of Proposition \ref{prop:pspec}

		\appendix
		\section{A result on radial functions}
		\begin{lemma}\label{lem:radfun}
			Let $a\in L^2((0,+\infty),rdr)$ be a function such that $\dot{a}, \frac{a}{r}\in\break L^2((0,+\infty),rdr)$. We have
			\[
				\RE\int_{0}^\infty \left(\overline{\dot{a}(r)}\frac{a(r)}{r}\right)rdr = 0.
			\]
		\end{lemma}
		\begin{proof}
			 The function $r^{1/2}a$ belongs to $H^1(0,+\infty)$ thus, in particular, $a\in C^0(0,\infty)$. 	
			For $r_0>0$, we have
			\[
				\RE\int_{r_0}^\infty \left(\overline{\dot{a}(r)}\frac{a(r)}{r}\right)rdr = \int_{r_0}^\infty \frac{d}{dr}|a|^2(r)dr = -|a|^2(r_0).
			\]
			Hence $|a|^2$ has a finite limit at $0$. Since $a/r\in L^2((0,+\infty),rdr)$, we get $|a|^2(0) = 0$ which yields
			\[
				\RE\int_{0}^\infty \left(\overline{\dot{a}(r)}\frac{a(r)}{r}\right)rdr = 0.
			\]

		\end{proof}

\section*{Acknowledgments}
 L. L.T. wishes to thank J. Olivier for various discussions on special functions and for pointing out the related references cited in this paper.
\\
 T. O.-B. wishes to thank K. Pankrashkin for useful discussions and remarks about self-adjoint extensions of symmetric operators.
\subsection*{Fundings}
T. O.-B. is supported by a public grant as part of the {\it Investissement d'avenir} project, reference ANR-11-LABX-0056-LMH, LabEx LMH.
\\
L.~L.T. is supported by ANR DYRAQ ANR-17-CE40-0016-01.
	
	\bibliographystyle{abbrv}

\end{document}